\documentclass[a4paper,10pt]{article}
\usepackage{graphicx}

\newcommand{\eps}{\epsilon}

\newcommand{\ignore}[1]{}

\newcommand{\qed}{\hfill\rule{.5em}{1.5ex}}
\newenvironment{proof}{\begin{trivlist}
    \item[] {\bf Proof:}}{\hspace{1.5em}\qed\end{trivlist}}

\newcommand{\eqref}[1]{{\rm(\ref{#1})}}

\newenvironment{smallenumerate}%
   {\vspace*{-4pt}
    \begin{enumerate}\itemsep=0pt}%
   {\end{enumerate}
    \vspace*{-2pt}}

   {\vspace*{-4pt}
    \begin{itemize}\itemsep=0pt}%
   {\end{itemize}
    \vspace*{-2pt}}

\newtheorem{theorem}{Theorem}

\newtheorem{lemma}{Lemma}
\newtheorem{fact}{Fact}

\newcommand{\hk}[1]{H_{#1}}     
\newcommand{\hka}[1]{H_{#1} ^a}     

\newcommand{\MTF} {{\tt MTF}}
\newcommand{\RLE} {{\tt RLE}}
\newcommand{\BWT} {{\tt BWT}}
\newcommand{\RBWT} {{\tt RBWT}}
\newcommand{\XBW} {{\tt XBW}}
\newcommand{\rem} {{\tt Remove}}
\newcommand{\app} {{\tt Append}}
\newcommand{\prev} {{\tt Prev}}
\newcommand{\next} {{\tt Next}}
\newcommand{\len} {{\tt Size}}
\newcommand{\C} {{\cal C}}
\newcommand{\G} {{\cal G}}
\newcommand{\Part} {{\cal P}}
\newcommand{\POpt} {{{\cal P}}_{\tt opt}}
\newcommand{\cost} {{\tt Cost}}
\newcommand{\PG} {{{\cal G}}_\epsilon}
\newcommand{\PPC} {{\sc PPC}}
\newcommand{\gzip} {{\tt gzip}}
\newcommand{\bzip} {{\tt bzip2}}
\newcommand{\ppm} {{\tt ppm}}

\begin{document}
\title{On optimally partitioning a text to improve its compression
\footnote{It has been partially supported by {\sf Yahoo! Research}, Italian MIUR Italy-Israel FIRB Project. The authors' address is Dipartimento di Informatica, L.go B. Pontecorvo 3, 56127 Pisa, Italy.}}

\author{Paolo Ferragina \and Igor Nitto \and Rossano Venturini \footnote{Department of Computer Science, University of Pisa.
emails: {\tt \{ferragina, nitto, rventurini\}@di.unipi.it}} }

\date{}

\maketitle

\begin{abstract}
In this paper we investigate the problem of partitioning an input string $T$ in such 
a way that compressing individually its parts via a base-compressor $\C$ gets a compressed 
output that is shorter than applying $\C$ over the entire $T$ at once. This problem was 
introduced in \cite{buchsbaum,Tables} in the context of table compression, and then further 
elaborated and extended to strings and trees by \cite{FGMS05,xbw05,MN07impl}. Unfortunately,
the literature offers poor solutions: namely, we know either a cubic-time algorithm for 
computing the optimal partition based on dynamic programming \cite{Tables,GS03}, or few 
heuristics that do not guarantee any bounds on the efficacy of their computed partition 
\cite{buchsbaum,Tables}, or algorithms that are efficient but work in some specific scenarios 
(such as the Burrows-Wheeler Transform, see e.g. \cite{FGMS05,MN07impl}) and achieve compression 
performance that might be worse than the optimal-partitioning by a $\Omega(\sqrt{\log n})$ factor. 
Therefore, computing efficiently the optimal solution is still open \cite{Giancarlo-Enc}. In this paper 
we provide the first algorithm which is guaranteed to compute in $O(n \log_{1+\eps}n)$ time a partition 
of $T$ whose compressed output is guaranteed to be no more than $(1+\epsilon)$-worse the optimal one, 
where $\epsilon$ may be any positive constant.
\end{abstract}

\section{Introduction}
\label{sec:intro}

Reorganizing data in order to improve the performance of a given compressor $\C$ is a recent 
and important paradigm in data compression (see e.g. \cite{Tables,FGMS05}). The basic idea 
consist of {\em permuting} the input data $T$ to form a new string $T'$ which is then {\em partitioned} 
into substrings $T' = T'_1 T'_2 \cdots T'_k$ that are finally compressed {\em individually} 
by the base compressor $\C$. The goal is to find the best instantiation of the two steps 
Permuting+Partitioning so that the compression of the individual substrings $T'_i$ minimizes 
the total length of the compressed output. This approach (hereafter abbreviated as \PPC) is 
clearly {\em at least} as powerful as the classic data compression approach that applies $\C$ 
to the entire $T$: just take the identity permutation and set $k=1$. The question is whether 
it can be {\em more powerful} than that!

Intuition leads to think favorably about it: by grouping together objects that are ``related'', 
one can hope to obtain better compression even using a very weak compressor $\C$. Surprisingly 
enough, this intuition has been sustained by convincing theoretical and experimental results only 
recently. These results have investigated the \PPC-paradigm under various angles by considering: 
different data formats (strings \cite{FGMS05}, trees \cite{xbw05}, tables \cite{Tables}, etc.), 
different granularities for the items of $T$ to be permuted (chars, node labels, columns, 
blocks \cite{BentleyM01,KulkarniDLT04}, files \cite{BigTable08,SuelM02,SuelTSP}, etc.), different 
permutations (see e.g. \cite{GSR07,VoV07tcs,SuelTSP,BigTable08}), different base compressors to 
be boosted (0-th order compressors, \gzip, \bzip, etc.). Among these plethora of proposals, we 
survey below the most notable examples which are useful to introduce the problem we attack in 
this paper, and refer the reader to the cited bibliography for other interesting results.

The \PPC-paradigm was introduced in \cite{buchsbaum}, and further
elaborated upon in \cite{Tables}. In these papers $T$ is a {\em
table} formed by fixed size columns, and the goal is to permute the
columns in such a way that individually compressing contiguous
groups of them gives the shortest compressed output. The authors of
\cite{Tables} showed that the \PPC-problem in its full generality is
MAX-SNP hard, devised a link between \PPC\ and the classical
asymmetric TSP problem, and then resorted known {\em heuristics} to
find approximate solutions based on several measures of correlations
between the table's columns. For the grouping they proposed either
an optimal but very slow approach, based on Dynamic Programming (see
below), or some very simple and fast algorithms which however did
not have any guaranteed bounds in terms of efficacy of their
grouping process. Experiments showed that these heuristics achieve
significant improvements over the classic \gzip, when it is applied
on the serialized original $T$ (row- or column-wise). Furthermore,
they showed that the combination of the TSP-heuristic with the
DP-optimal partitioning is even better, but it is too slow to be
used in practice even on small file sizes because of the DP-cubic
time complexity.\footnote{Page 836 of \cite{Tables} says: "computing
a good approximation to the TSP reordering before partitioning
contributes significant compression improvement at minimal time
cost. [...] This time is negligible compared to the time to compute
the optimal, contiguous partition via DP."}

When $T$ is a text string, the most famous instantiation of the
\PPC-paradigm has been obtained by combining the Burrows and Wheeler
Transform~\cite{bw} (shortly \BWT) with a context-based grouping of
the input characters, which are finally compressed via proper $0$-th
order-entropy compressors (like \MTF, \RLE, Huffman, Arithmetic, or
their combinations, see e.g. \cite{Witten:1999:MGC}). Here the
\PPC-paradigm takes the name of {\em compression booster}
\cite{FGMS05} because the net result it produces is to boost the
performance of the base compressor $\C$ from $0$-th order-entropy
bounds to $k$-th order entropy bounds, simultaneously over all $k
\geq 0$. In this scenario the permutation acts on single characters,
and the partitioning/permuting steps deploy the context (substring)
following each symbol in the original string in order to identify
``related'' characters which must be therefore compressed together.
Recently \cite{GSR07} investigated whether do exist other
permutations of the characters of $T$ which admit effective
compression and can be computed/inverted fast. Unfortunately they
found a connection between table compression and the \BWT, so that
many natural similarity-functions between contexts turned out to
induce MAX-SNP hard permuting problems! Interesting enough, the
\BWT\ seems to be the unique highly compressible permutation which
is fast to be computed and achieves effective compression bounds.
Several other papers have given an analytic account of this
phenomenon \cite{Manz01,FGM06,KLV07,MN07impl} and have shown, also
experimentally \cite{FGM06b}, that the partitioning of the
BW-transformed data is a key step for achieving effective
compression ratios. Optimal partitioning is actually even more
mandatory in the context of labeled-tree compression where a
BWT-inspired transform, called {\sf XBW}-transform in
\cite{xbw05,xbw-exp06}, allows to produce permuted strings with a
strong clustering effect. Starting from these premises
\cite{GS03} attacked the computation of the optimal
partitioning of $T$ via a DP-approach, which turned to be very
costly; then \cite{FGMS05} (and subsequently many other authors, see
e.g. \cite{FGM06,MN07impl,xbw05}) proposed solutions which are not
optimal but, nonetheless, achieve interesting $k$-th order-entropy
bounds. This
is indeed a subtle point which is frequently neglected when dealing
with compression boosters, especially in practice, and for this
reason we detail it more clearly in Appendix A in which we show an
infinite class of strings for which the compression achieved by the
booster is far from the optimal-partitioning by a multiplicative
factor $\Omega(\sqrt{\log n})$.

Finally, there is another scenario in which the computation of the
optimal partition of an input string for compression boosting can be
successful and occurs when $T$ is a single (possibly long) file on
which we wish to apply classic data compressors, such as \gzip,
\bzip, \ppm, etc. \cite{Witten:1999:MGC}. Note that how much
redundancy can be detected and exploited by these compressors
depends on their ability to ``look back'' at the previously seen
data. However, such ability has a cost in terms of memory usage and
running time, and thus most compression systems provide a facility
that controls the amount of data that may be processed at once ---
usually called the {\em block size}. For example the classic tools
\gzip\ and \bzip\ have been designed to have a small memory
footprint, up to few hundreds KBs. More recent and sophisticated
compressors, like \ppm~\cite{Witten:1999:MGC} and the family of
\BWT-based compressors~\cite{FGM06b}, have been designed to use
block sizes of up to a few hundreds MBs. But using larger blocks to
be compressed at once does not necessarily induce a better
compression ratio! As an example, let us take $\C$ as the simple
Huffman or Arithmetic coders and use them to compress the text $T =
0^{n/2}1^{n/2}$: There is a clear difference whether we compress
individually the two halves of $T$ (achieving an output size of
about $O(\log n)$ bits) or we compress $T$ as a whole (achieving $n
+ O(\log n)$ bits). The impact of the block size is even more
significant as we use more powerful compressors, such as the $k$-th
order entropy encoder $\ppm$ which compresses each symbol according
to its preceding $k$-long context. In this case take
$T=(2^k0)^{n/(2(k+1))} (2^k1)^{n/(2(k+1))}$ and observe that if we
divide $T$ in two halves and compress them individually, the output
size is about $O(\log n)$ bits, but if we compress the entire $T$ at
once then the output size turns to be much longer, i.e.
$\frac{n}{k+1} + O(\log n)$ bits. Therefore the choice of the block
size cannot be underestimated and, additionally, it is made even
more problematic by the fact that it is not necessarily the same
along the whole file we are compressing because it depends on the
distribution of the repetitions within it. This problem is even more
challenging when $T$ is obtained by concatenating a collection of
files via any permutation of them: think to the serialization
induced by the Unix {\tt tar} command, or other more sophisticated
heuristics like the ones discussed in
\cite{SuelM02,BigTable08,OuyangMST02,SuelTSP}. In these cases, the
partitioning step looks for {\em homogeneous} groups of contiguous
files which can be effectively compressed together by the
base-compressor $\C$. More than before, taking the largest
memory-footprint offered by $\C$ to group the files and compress
them at once is not necessarily the best choice because real
collections are typically formed by homogeneous groups of
dramatically different sizes (e.g. think to a Web collection and its
different kinds of pages). Again, in all those cases we could apply
the optimal DP-based partitioning approach of
\cite{GS03,Tables}, but this would take more than cubic time
(in the overall input size $|T|$) thus resulting unusable even on
small input data of few MBs!

\smallskip
In summary the efficient computation of an optimal partitioning of
the input text for compression boosting is an important and still
open problem of data compression (see \cite{Giancarlo-Enc}). The
goal of this paper is to make a step forward by providing the first
efficient approximation algorithm for this problem, formally stated
as follows.

Let $\C$ be the base compressor we wish to boost, and let $T[1,n]$
be the input string we wish to partition and then compress by $\C$.
So, we are assuming that $T$ has been (possibly) permuted in
advance, and we are concentrating on the last two steps of the
\PPC-paradigm. Now, given a partition $\Part$ of the input string
into contiguous substrings, say $T = T_1 T_2 \cdots T_k$, we denote
by $\cost(\Part)$ the cost of this partition and measure it as
$\sum_{i=1}^l |\C(T_i)|$, where $|\C(\alpha)|$ is the length in bit
of the string $\alpha$ compressed by $\C$. The problem of {\em
optimally partitioning} $T$ according to the base-compressor $\C$
consists then of computing the partition $\POpt$ achieving the
minimum cost, namely $\POpt = \min_{\Part} \cost(\Part)$, and thus
the shortest compressed output.\footnote{We are assuming that
$\C(\alpha)$ is a prefix-free encoding of $\alpha$, so that we can
concatenate the compressed output of many substrings and still be
able to recover them via a sequential scan.}

As we mentioned above $\POpt$ might be computed via a
Dynamic-Programming approach \cite{Tables,GS03}. Define $E[i]$ as
the cost of the optimum partitioning of $T[1,i]$, and set $E[0] =
0$. Then, for each $i \geq 1$, we can compute $E[i]$ as the
$min_{0\leq j \leq i-1} E[j] + |\C(T[j+1,i])|$. At the end $E[n]$
gives the cost of $\POpt$, which can be explicitly determined by
standard back-tracking over the DP-array. Unfortunately, this
solution requires to run $\C$ over $\Theta(n^2)$ substrings of
average length $\Theta(n)$, for an overall $\Theta(n^3)$ time cost
in the worst case which is clearly unfeasible even on small input
sizes $n$.

In order to overcome this computational bottleneck we make two
crucial observations: (1) instead of applying $\C$ over each
substring of $T$, we use an entropy-based estimation of $\C$'s
compressed output that can be computed efficiently and incrementally
by suitable dynamic data structures; (2) we relax the requirement
for an exact solution to the optimal partitioning problem, and aim
at finding a partition whose cost is no more than $(1+\eps)$ worse
than $\POpt$, where $\eps$ may be any positive constant. Item (1)
takes inspiration from the heuristics proposed in
\cite{buchsbaum,Tables}, but it is executed in a more principled way
because our entropy-based cost functions reflect the real behavior
of modern compressors, and our dynamic data structures allow the
efficient estimation of those costs without their re-computation
from scratch at each substring (as instead occurred in
\cite{buchsbaum,Tables}). Item (2) boils down to show that the
optimal partitioning problem can be rephrased as a Single Source
Shortest path computation over a weighted {\tt DAG} consisting of
$n$ nodes and $O(n^2)$ edges whose costs are derived from item (1).
We prove some interesting structural properties of this graph that
allow us to restrict the computation of that {\tt SSSP} to a
subgraph consisting of $O(n\log_{1+\eps} n)$ edges only. The
technical part of this paper (see Section \ref{sec:approx}) will
show that we can build this graph on-the-fly as the {\tt
SSSP}-computation proceeds over the {\tt DAG} via the proper use of
time-space efficient dynamic data structures. The final result will
be to show that we can $(1+\eps)$-approximate $\POpt$ in
$O(n\log_{1+\eps} n)$ time and $O(n)$ space, for both $0$-th order
compressors (like Huffman and Arithmetic \cite{Witten:1999:MGC}) and
$k$-th order compressors (like \ppm\ \cite{Witten:1999:MGC}). We
will also extend these results to the class of \BWT-based
compressors, when $T$ is a collection of texts.

We point out that the result on $0$-th order compressors is
interesting in its own from both the experimental side, since {\em
Huffword} compressor is the standard choice for the storage of Web
pages \cite{Witten:1999:MGC}, and from the theoretical side since it
can be applied to the compression booster of \cite{FGMS05} to fast
obtain an approximation of the optimal partition of $\BWT(T)$ in
$O(n\log_{1+\eps} n)$ time. This may be better than the algorithm of
\cite{FGMS05} both in time complexity, since that takes $O(n
|\Sigma|)$ time where $\Sigma$ is the alphabet of $T$, and in
compression ratio (as we have shown above, see Appendix A). The case
of a large alphabet (namely, $|\Sigma| = \Omega({\tt polylog}(n))$) is
 particularly interesting whenever we consider either a
word-based \BWT\ \cite{MI05} or the \XBW-transform over labeled
trees \cite{FGMS05}. Finally, we mention that our results apply also
to the practical case in which the base compressor $\C$ has a
maximum (block) size $B$ of data it can process at once (see above
the case of \gzip, \bzip, etc.). In this situation the time
performance of our solution reduces to $O(n\log_{1+\eps} (B \log
\sigma))$.

The map of the paper is as follows. Section \ref{sec:notation}
introduces some basic notation and terminology. Section
\ref{sec:approx} describes our reduction from the optimal
partitioning problem of $T$ to a {\tt SSSP} problem over a weighted
{\tt DAG} in which edges represent substrings of $T$ and edge costs
are entropy-based estimations of the compression of these substrings
via $\C$. The subsequent Sections will address the problem of
incrementally and efficiently computing those edge costs as they are
needed by the {\tt SSSP}-computation, distinguishing the two cases
of $0$-th order estimators (Section \ref{sec:h0}) and $k$-th order
estimators (Section \ref{sec:hk}), and the situation in which $\C$ is
a \BWT-based compressor and $T$ is a collection of files (Section
\ref{sec:bwt}).

\section{Notation} \label{sec:notation}
In this paper we will use entropy-based upper bounds for the
estimation of $|\C(T[i,j])|$, so we need to recall some basic
notation and terminology about entropies. Let $T[1,n]$ be a string
drawn from the alphabet $\Sigma$ of size $\sigma$. For each $c \in
\Sigma$, we let $n_c$ be the number of occurrences of $c$ in $T$.
The zero-th order {\it empirical} entropy of $T$ is defined as
$\hk{0}(T) = \frac{1}{|T|}\sum_{c \in  \Sigma}^h n_c \: \log
\frac{n}{n_c}$.

Recall that $|T| \hk{0}(T)$ provides an information-theoretic lower
bound to the output size of any compressor that encodes each symbol
of $T$ with a fixed code \cite{Witten:1999:MGC}. The so-called
zero-th order statistical compressors (such as Huffman or Arithmetic
\cite{Witten:1999:MGC}) achieve an output size which is very close
to this bound. However, they require to know information about
frequencies of input symbols (called the {\em model} of the source).
Those frequencies can be either known in advance ({\em static}
model) or computed by scanning the input text ({\em semistatic}
model). In both cases the model must be stored in the compressed
file to be used by the decompressor.

In the following we will bound the compressed size achieved by zero-th
order compressors over $T$ by $|\C_0(T)| \leq \lambda n \hk{0}(T) + f_0(n, \sigma)$ bits,
where $\lambda$ is a positive constant and $f_0(n,\sigma)$ is a function 
including the extra costs of encoding the source
model and/or other inefficiencies of $\C$.
In the following we will assume that the function $f_0(n, \sigma)$ can be
computed in constant time given $n$ and $\sigma$.
As an example, for Huffman $f_0(n, \sigma) = \sigma \log \sigma +n$ bits and $\lambda=1$, and
for Arithmetic $f_0(n, \sigma) = \sigma \log n + \log n / n$ bits and $\lambda = 1$.

In order to evict the cost of the model, we can resort to zero-th order
\textit{adaptive} compressors that do not require to know the symbols' frequencies in advance, since
they are computed incrementally during the compression. The zero-th order
 {\it adaptive empirical} entropy of $T$ \cite{HowardV92} is then defined 
as $\hk{0}^a(T) = \frac{1}{|T|}\sum_{c \in  \Sigma}^h \log \frac{n!}{n_c!}$
We will bound the compress size achieved by zero-th order
adaptive compressors over $T$ by $|\C_0^a(T)| = n \hk{0}^a(T)$ bits.

Let us now come to more powerful compressors. For any string $u$ of
length $k$, we denote by $u_T$ the string of single symbols
following the occurrences of $u$ in~$T$, taken from left to right.
For example, if $T={\tt mississippi}$ and $u={\tt si}$, we have $u_T
= {\tt sp}$ since the two occurrences of {\tt si} in $T$ are
followed by the symbols {\tt s} and {\tt p}, respectively. The
$k$-th order {\em empirical} entropy of $T$ is defined as $\hk{k}(T)
= \frac{1}{|T|} \sum_{u\in \Sigma^{k}} |u_T|\: \hk{0}(u_T)$. Analogously, 
the $k$-th order {\em adaptive empirical} entropy of $T$ is defined as 
$\hk{k}^a(T) = \frac{1}{|T|} \sum_{u\in \Sigma^{k}} |u_T|\:\hk{0}^a(u_T)$

We have $H_k(T) \geq H_{k+1}(T)$ for any $k\geq 0$. As usual in
data compression \cite{Manz01}, the value $n H_k(T)$ is an
information-theoretic lower bound to the output size of any
compressor that encodes each symbol of $T$ with a fixed code that
depends on the symbol itself and on the $k$ immediately preceding
symbols.  Recently (see e.g.
\cite{koma00,Manz01,FGMS05,FGM06,MN07impl,xbw05} and refs therein)
authors have provided {\em upper bounds} in terms of $H_k(|T|)$ for
sophisticated data-compression algorithms, such as \gzip\
\cite{koma00}, \bzip\ \cite{Manz01,FGMS05,KLV07}, and \ppm. These
bounds have the form $|\C(T)| \leq \lambda |T| \: H_k(T) + f_k(|T|,
\sigma)$, where $\lambda$ is a positive constant and $f_k(|T|,
\sigma)$ is a function including the extra-cost of encoding the
source model and/or other inefficiencies of $\C$. The smaller are
$\lambda$ and $f_k()$, the better is the compressor $\C$. As an
example, the bound of the compressor in \cite{MN07impl} has $\lambda=1$ and $f(|T|,
\sigma) = O(\sigma^{k+1} \log |T| + |T| \log \sigma \log \log |T|/ \log |T|)$.
Similar bounds that involve the adaptive $k$-th order entropy are 
known \cite{Manz01,FGMS05,FGM06} for many compressors. In these cases 
the bound takes the form $|\C_k ^a(T)| = \lambda|T|\hk{k}^*(T)+ g_k(\sigma)$ 
bits, where the value of $g_k$ depends only on the alphabet size $\sigma$.

In our paper we will use these entropy-based bounds for the
estimation of $|\C(T[i,j])|$, but of course this will not be
enough to achieve a fast DP-based algorithm for our optimal-partitioning
problem. We cannot re-compute from scratch those estimates
for every substring $T[i,j]$ of $T$, being them $\Theta(n^2)$ in number. So we will show some
structural properties of our problem (Section \ref{sec:approx}) and
introduce few novel technicalities (Sections
\ref{sec:h0}--\ref{sec:hk}) that will allow us to compute
$H_k(T[i,j])$ only on a {\em reduced} subset of $T$'s substrings, having size $O(n \log_{1+\epsilon} n)$, by taking
$O({\tt polylog}(n))$ time per substring and $O(n)$ space overall.

\section{The problem and our solution} \label{sec:approx}

The optimal partitioning problem, stated in Section
\ref{sec:intro} can be reduced to a single source shortest path
computation (SSSP) over a directed acyclic graph $\G(T)$
defined as follows. The graph $\G(T)$ has a vertex $v_i$
for each text position $i$ of $T$, plus an additional vertex
$v_{n+1}$ marking the end of the text, and an edge connecting vertex
$v_i$ to vertex $v_j$ for any pair of indices $i$ and $j$ such that
$i < j$. Each edge $(v_i,v_j)$ has associated the cost $c(v_i,v_j) =
|\C(T[i,j-1])|$ that corresponds to the size in bits of the
substring $T[i,j-1]$ compressed by $\C$. We remark the following
crucial, but easy to prove, property of the cost function defined on
$\G(T)$:

\begin{fact}\label{prop:increasing}
 For any vertex $v_i$, it is $0 < c(v_i,v_{i+1}) \leq c(v_i,v_{i+2}) \leq \ldots \leq c(v_i,v_{n+1})$
\end{fact}

There is a one-to-one correspondence between paths from $v_1$ to $v_{n+1}$ in $\G(T)$ 
and partitions of $T$: every edge $(v_i,v_j)$ in the path identifies a contiguous 
substring $T[i,j-1]$ of the corresponding partition. Therefore the cost of a path 
is equal to the (compression-)cost of the
corresponding partition. Thus, we can find the optimal partition of $T$  by computing 
the shortest path in $\G(T)$ from $v_1$ to $v_{n+1}$.
Unfortunately this simple approach has two main drawbacks:
\begin{smallenumerate}
\item the number of edges in $\G(T)$ is $\Theta(n^2)$, thus making the SSSP computation
inefficient (i.e. $\Omega(n^2)$ time) if executed directly over
$\G(T)$;

\item the computation of the each edge cost might take $\Theta(n)$ time over most $T$'s substrings,
if $\C$ is run on each of them from scratch.
\end{smallenumerate}

In the following sections we will successfully address  both these two
drawbacks. First, we sensibly reduce the number of edges in the
graph $\G(T)$ to be examined during the SSSP computation and show
that we can obtain a $(1+\epsilon)$ approximation using only $O(n
\log_{1+\epsilon} n)$ edges, where $\epsilon > 0$ is a user-defined
parameter (Section \ref{sub:pruning}). Second, we show some sufficient
properties that $\C$ needs to satisfy in order to compute
efficiently every edge's cost. These properties hold for some
well-known compressors--- e.g. $0$-order compressors,
\texttt{PPM}-like and \texttt{bzip}-like compressors--- and for them
we show how to compute each edge cost in constant or polylogarithmic
time (Sections \ref{sec:h0}---\ref{sec:bwt}).

\subsection{A pruning strategy}
\label{sub:pruning}

The aim of this section is to design a {\em pruning} strategy that produces a subgraph $\PG(T)$ of 
the original DAG $\G(T)$ in which the shortest
path distance between its leftmost and rightmost nodes, $v_1$ and
$v_{n+1}$, increases by no more than a factor $(1+\eps)$.
We define $\PG(T)$ to contain all edges $(v_i,v_j)$ of $\G(T)$, recall $i
< j$, such that at least one of the following two conditions holds:

\begin{smallenumerate}
 \item there exists a positive integer $k$ such that $c(v_i,v_j) \leq (1+\eps)^k < c(v_i,v_{j+1})$;
 \item $j = n+1$.
\end{smallenumerate}

In other words, by fact \ref{prop:increasing}, we are keeping for
each integer $k$ the edge of $\G(T)$ that approximates at the best
the value $(1+\eps)^k$ from below. Given this, we will call {\em
$\eps$-maximal} the edges of $\PG(T)$. Clearly, each vertex of
$\PG(T)$ has at most $\log_{1+\eps}n = O(\frac{1}{\eps}\log n)$
outgoing edges, which are $\eps$-maximal by definition. Therefore
the total size of $\PG(T)$ is at most $O(\frac{n}{\eps}\log n)$.
Hereafter, we will denote with $d_G(-,-)$ the shortest path distance
between any two nodes in a graph $G$.

The following lemma states a basic property of shortest path
distances over our special DAG $\G(T)$:

\begin{lemma}\label{lem:dist}
For any triple of indices $1 \leq i \leq j \leq q \leq n+1$ we have:
 \begin{enumerate}
  \item $d_{\G(T)}(v_j,v_q) \leq d_{\G(T)}(v_i,v_q)$
  \item $d_{\G(T)}(v_i,v_j) \leq d_{\G(T)}(v_i,v_q)$
 \end{enumerate}
\end{lemma}
\begin{proof}
 We prove just 1, since 2 is symmetric. It suffices by induction to prove the 
case $j=i+1$. Let $(v_i, w_1)(w_1, w_2)...(w_{h-1},w_h)$, with $w_h = v_q$, be 
a shortest path in $\G(T)$ from $v_i$ to $v_q$. By fact \ref{prop:increasing}, 
$c(v_j,w_1) \leq c(v_i,w_1)$ since $i \leq j$. Therefore the cost of the path
$(v_j, w_1)(w_1, w_2)...(w_{h-1},w_h)$ is at most $d_{\G(T)}(v_i,v_q)$, which 
proves the claim.
\end{proof}

\noindent
The correctness of our pruning strategy relies on the following theorem:

\begin{theorem}\label{teo:pruning}
For any text $T$, the shortest path in $\PG(T)$ from $v_1$ to $v_{n+1}$ has 
a total cost of at most $(1+\eps)\; d_{\G(T)}(v_1,v_{n+1})$.
\end{theorem}\begin{proof}
We prove a stronger assertion: $d_{\PG(T)}(v_i,v_{n+1}) \leq (1+\eps)\; d_{\G(T)}(v_i,v_{n+1})$ for
any index $1 \leq i \leq n+1$.
This is clearly true for $i=n+1$, because in that case the distance is 0. Now let us inductively consider the shortest path $\pi$ in $\G(T)$ from $v_i$ to $v_{n+1}$ and
let $(v_{k},v_{t_1})(v_{t_1},v_{t_2}) \dots (v_{t_{h}}v_{n+1})$
be its edges. By the definition of $\eps$-maximal edge,
it is possible to find an $\eps$-maximal edge $(v_k,v_r)$ with $t_1 \leq r$, such that $c(v_k,v_{r}) \leq (1+\eps)\; c(v_k,v_{t_1})$. By Lemma \ref{lem:dist},
$d_{\G(T)}(v_r,v_{n+1}) \leq d_{\G(T)}(v_{t_1},v_{n+1})$. By induction,
$d_{\PG(T)}(v_{r}, v_{n+1}) \leq (1+\eps)\; d_{\G(T)}(v_r,v_{n+1})$. Combining this with the
 triangle inequality we get the thesis.\end{proof}

\subsection{Space and time efficient algorithms for generating $\PG(T)$}

Theorem \ref{teo:pruning} ensures that, in order to compute a
$(1+\eps)$ approximation of the optimal partition of $T$, it
suffices to compute the SSSP in $\PG(T)$ from $v_1$ to $v_{n+1}$.
This can be easily computed in $O(|\PG(T)|) = O(n\log_{\eps}
n)$ time since $\PG(T)$ is a DAG \cite{CLR}, by making a single pass
over its vertices and relaxing all edges going out from the current
one.

However, generating $\PG(T)$ in efficient time is a non-trivial task
for three main reasons. First, the original graph $\G(T)$
contains $\Omega(n^2)$ edges, so that we cannot check each
of them to determine whether it is $\eps$-maximal or not, because
this would take $\Omega(n^2)$ time. Second, we cannot compute the cost of an edge
$(v_i,v_j)$ by executing $\C(T[i,j-1])$ {\em from scratch}, since
this would require time linear in the substring length, and thus
$\Omega(n^3)$ time over all $T$'s substrings. Third, we cannot
materialize $\PG(T)$ (e.g. its adjacency lists) because it consists
of $\Theta(n \; {\tt polylog}(n))$ edges, and thus its space
occupancy would be super-linear in the input size.

The rest of this section is devoted to design an algorithm which
overcomes the three limitations above. The specialty of our
algorithm consists of materializing $\PG(T)$ on-the-fly, as its
vertices are examined during the SSSP-computation, by spending only
polylogarithmic time per edge. The actual time complexity per edge
will depend on the entropy-based cost function we will use to
estimate $|\C(T[i,j-1])|$ (see Section \ref{sec:notation}) and on
the dynamic data structure we will deploy to compute that estimation
efficiently.

The key tool we use to make a fast estimation of the edge costs
is a dynamic data structure built over the input text $T$ and
requiring $O(|T|)$ space. We state the main properties of this
data structure in an abstract form, in order to design a general framework for
solving our problem; in the next sections we will then provide
implementations of this data structure and thus obtain real
time/space bounds for our problem. So, let us assume to have a
dynamic data structure that maintains a set of {\em sliding windows}
over $T$ denoted by $w_1, w_2, \dots, w_{\log_{1+\epsilon} n}$. The
sliding windows are substrings of $T$ which start at the same text
position $l$ but have different lengths: namely, $w_i = T[l,r_i]$
and $r_1 \leq r_2 \leq \ldots \leq r_{\log_{1+\epsilon} n}$. The
data structure must support the following three operations:

\begin{smallenumerate}
 \item $\rem()$ moves the starting position $l$ of all windows one position to the right (i.e. $l+1$);
 \item $\app(w_i)$ moves the ending position of the window $w_i$ one position to the right (i.e. $r_i+1$);
 \item $\len(w_i)$ computes and returns the value $|\C(T[l,r_i])|$.
\end{smallenumerate}

This data structure is enough to generate $\eps$-maximal edges via a
single pass over $T$, using $O(|T|)$ space. More precisely, let
$v_l$ be the vertex of $\G(T)$ currently examined by our SSSP
computation, and thus $l$ is the current position reached by our
scan of $T$. We maintain the following invariant: the sliding windows
correspond to all $\eps$-maximal edges going out from $v_l$, that is, the edge
$(v_l,v_{1+r_t})$ is the $\eps$-maximal edge satisfying
$c(v_l,v_{1+r_t}) \leq (1+\eps)^t < c(v_l,v_{1+(r_{t}+1)})$.
Initially all indices are set to $0$.
To maintain the invariant, when the text scan advances to the
next position $l+1$, we call operation $\rem()$ once to increment
index $l$ and, for each $t=1,\ldots, \log_{1+\eps} (n)$, we call
operation $\app(w_t)$ until we find the largest $r_t$ such that
$\len(w_t) = c(v_l,v_{1+r_t}) \leq (1+\eps)^t$. The key issue here
is that $\app$ and $\len$ are paired so that our data structure
should take advantage of the rightward sliding of $r_t$ for
computing $c(v_l,v_{1+r_t})$ efficiently. Just one character is
entering $w_t$ to its right, so we need to deploy this fact for
making the computation of $\len(w_t)$ fast (given its previous
value). Here comes into play the second contribution of our paper
that consists of adopting the entropy-bounded estimates for the
compressibility of a string, mentioned in Section \ref{sec:notation},
to estimate indeed the edge costs $\len(w_t)=|\C(w_t)|$. This idea
is crucial because we will be able to show that these functions do
satisfy some structural properties that admit a {\em fast
incremental computation}, as the one required by $\app + \len$.
These issues will be discussed in the following sections, here we
just state that, overall, the SSSP computation over $\PG(T)$ takes
$O(n)$ calls to operation $\rem$, and $O(n \log_{1+\eps} n)$ calls to
operations $\app$ and $\len$.

\begin{theorem} \label{th:final}
   If we have a dynamic data structure occupying $O(n)$ space and supporting operation $\rem$ in time $L(n)$, 
and operations $\app$ and $\len$ in time $R(n)$, then we can compute the shortest path in $\PG(T)$ from
 $v_1$ to $v_{n+1}$ taking $O(n\; L(n) + (n \log_{1+\eps}n)\; R(n))$ time and $O(n)$ space.
\end{theorem}

\section{On zero-th order compressors}
\label{sec:h0}

In this section we explain how to implement the data structure above whenever $\C$ is a $0$-th order compressor, 
and thus $H_0$ is used to provide a bound to the compression cost of $\G(T)$'s edges (see Section \ref{sec:notation}). 
The key point is actually to show how to efficiently compute $\len(w_i)$ as the sum 
of $|T[l,r_i]|\hk{0}(T[l,r_i])= \sum _{c \in \Sigma} n_c \log ((r_i-l+1)/n_c)$ (see 
its definition in Section \ref{sec:notation}) plus $f_0(r_i-l+1, |\Sigma_{T[l,r_i]}|)$,
where $n_c$ is the number of occurrences of symbol $c$ in $T[l,r_i]$ and $|\Sigma_{T[l,r_i]}|$
denotes the number of different symbols in $T[l,r_i]$.

The first solution we are going to present is very simple and uses $O(\sigma)$ space
per window. The idea is the following: for each window $w_i$ we keep in memory an array of counters $A_i[c]$
indexed by symbol $c$ in $\Sigma$.
At any step of our algorithm, the counter $A_i[c]$ stores the number of
occurrences of symbol $c$ in $T[l,r_i]$. For any window $w_i$, we also use a variable $E_i$ 
that stores the value of $\sum _{c \in \Sigma} A_i[c] \log A_i[c]$.
It is easy to notice that:

\begin{equation} \label{Ei}
|T[l,r_i]|\; \hk{0}(T[l,r_i]) = (r_i-l+1) \log (r_i-l+1) - E_i.
\end{equation}

Therefore, if we know the value of $E_i$, we can answer to a query
$\len(w_i)$ in constant time. So, we are left with showing how to
implement efficiently the two operations that modify $l$ or any $r$s
value and, thus, modify appropriately the $E$'s value. This can be
done as follows:

\begin{smallenumerate}
\item \rem$()$: For each window $w_i$, we subtract from the appropriate counter and from variable $E_i$
the contribution of the symbol $T[l]$ which has been evicted from the window.
That is, we decrease $A_i[T[l]]$ by one, and update $E_i$ by subtracting
$(A_i[T[l]]+1) \log (A_i[T[l]]+1)$ and then summing $A_i[T[l]] \log A_i[T[l]]$.
Finally we set $l=l+1$.

\item \app$(w_i)$:  We add to the appropriate counter and variable $E_i$ the contribution of the symbol $T[r_i+1]$
which has been appended to window $w_i$.
That is, we increase $A_i[T[r+1]]$ by one, then we update $E_i$ by subtracting 
$(A[T[r_i+1]]-1) \log (A[T[r_i+1]]-1)$ and summing $A[T[r_i+1]] \log A[T[r_i+1]]$. 
Finally we set $r_i=r_i+1$.
\end{smallenumerate}

In this way, operation $\rem$ requires constant time per window,
hence $O(\log_{1+\eps} n)$ time overall. $\app(w_i)$ takes constant
time. The space required by the counters $A_i$ is $O(\sigma
\log_{1+\eps} n)$ words. Unfortunately, the space complexity of this
solution can be too much when it is used as the basic-block for
computing the $k$-th order entropy of $T$  (see Section
\ref{sec:notation}) as we will do in Section \ref{sec:hk}. In fact,
we would achieve $\min( \sigma^{k+1} \log_{1+\eps} n , n
\log_{1+\eps} n )$ space, which may be superlinear in $n$ depending
on $\sigma$ and $k$.

The rest of this section is therefore devoted to provide an
implementation of our dynamic data structure that takes the same
query time above for these three operations, but within $O(n)$
space, which is independent of $\sigma$ and $k$. The new solution
still uses $E$'s value but the counters $A_i$ are computed
on-the-fly by exploiting the fact that all windows share the same
value of $l$. We keep an array $B$ indexed by symbols whose entry
$B[c]$ stores the number of occurrences of $c$ in $T[1,l]$. We can
keep these counters updated after a $\rem$ by simply decreasing
$B[T[l]]$ by one. We also maintain an array $R$ with an entry for
each text position. The entry $R[j]$ stores the number of
occurrences of symbol $T[j]$ in $T[1,j]$. The number of elements in
both $B$ and $R$ is no more than $n$, hence they take $O(n)$ space.

These two arrays are enough to correctly update the value $E_i$
after $\app(w_i)$, which is in turn enough to estimate $H_0$ (see
Eqn \ref{Ei}). In fact, we can compute the value $A_i[T[r_i+1]]$ by
computing $R[r_i+1] - B[T[r_i+1]]$ which correctly reports the
number of occurrences of $T[r_i+1]$ in $T[l \dots r_i+1]$. Once we
have the value of $A_i[T[r_i+1]]$, we can update $E_i$ as explained
in the above item $2$.

We are left with showing how to support $\rem()$ whose computation
requires to evaluate the value of $A_i[T[l]]$ for each window $w_i$.
Each of these values can be computed as $R[t] - B[T[l]]$ where $t$
is the last occurrence of symbol $T[l]$ in $T[l,r_i]$. The problem
here is given by the fact that we do not know the position $t$. We
solve this issue by resorting to a doubly linked list $L_c$ for each
symbol $c$. The list $L_c$ links together the last occurrences of
$c$ in all those windows, ordered by increasing position. Notice
that a position $j$ may be the last occurrence of symbol $T[j]$ for
different (but consecutive) windows. In this case we force that
position to occur in $L_{T[j]}$ just once. These lists are
sufficient to compute values $A_i[T[l]]$ for all the windows
together. In fact, since any position in $L_{T[l]}$ is the last
occurrence of at least one sliding window, each of them can be used
to compute $A_i[T[l]]$ for the appropriate indices $i$. Once we have
all values $A_i[T[l]]$, we can update all $E_i$'s as explained in
the above item $1$. Since list $L_{T[l]}$ contains no more than
$\log_{1+\eps} n$ elements, all $E$s can be updated in
$O(\log_{1+\eps} n)$ time. Notice that the number of elements in all
the lists $L$ is bounded by the text length. Thus, they are stored
using $O(n)$ space.

It remains to explain how to keep lists $L$ correctly updated.
Notice that only one list may change after a $\rem()$ or an
$\app(w_i)$. In the former case we have possibly to remove position
$l$ from list $L_{T[l]}$. This operation is simple because, if that
position is in the list, then $T[l]$ is the last occurrence of that
symbol in $w_1$ (recall that all the windows start at position $l$,
and are kept ordered by increasing ending position) and, thus, it
must be the head of $L_{T[l]}$. The case of $\app(w_i)$ is more
involved. Since the ending position of $w_i$ is moved to the right,
position $r_i+1$ becomes the last occurrence of symbol $T[r_i+1]$ in
$w_i$. Recall that $\app(w_i)$ inserts symbol $T[r_i+1]$ in $w_i$.
Thus, it must be inserted in $L_{T[r_i+1]}$ in its correct (sorted)
position, if it is not present yet. Obviously, we can do that in
$O(\log_{1+\eps} n)$ time by scanning the whole list. This is too
much, so we show how to spend only constant time. Let $p$ the
rightmost occurrence of the symbol $T[r_i+1]$ in $T[0,
r_i]$.\footnote{Notice that we can precompute and store the last
occurrence of symbol $T[j+1]$ in $T[1,j]$ for all $j$s in linear
time and space.} If $p < l$, then $r_i+1$ must be inserted in the
front of $L_{T[r_i+1]}$ and we have done. In fact, $p < l$ implies
that there is no occurrence of $T[r_i+1]$ in $T[l, r_i]$ and, thus,
no position can precede $r_i+1$ in $L_{T[r_i+1]}$. Otherwise (i.e.
$p \geq l$), we have that $p$ is in $L_{T[r_i+1]}$, because it is
the last occurrence of symbol $T[r_i+1]$ for some window $w_j$ with
$j\leq i$. We observe that if $w_j=w_i$, then $p$ must be replaced
by $r_i+1$ which is now the last occurrence of $T[r_i+1]$ in $w_i$;
otherwise $r_i+1$ must be inserted after $p$ in $L_{T[r_i+1]}$
because $p$ is still the last occurrence of this symbol in the
window $w_j$. We can decide which one is the correct case by
comparing $p$ and $r_{i-1}$ (i.e., the ending position of the
preceding window $w_{r_{i-1}}$). In any case, the list is kept
updated in constant time.

The following Lemma derives by the discussion above:

\begin{lemma}\label{lem:H0}
Let $T[1,n]$ be a text drawn from an alphabet of size $\sigma = {\tt poly}(n)$. 
If we estimate $\len()$ via $0$-th order
entropy (as detailed in Section \ref{sec:notation}), then we can design a 
dynamic data structure that takes $O(n)$ space and supports the operations
$\rem$ in $R(n)=O(\log_{1+\eps} n)$ time, and $\app$ and $\len$ in
$L(n)=O(1)$ time.
\end{lemma}

In order to evict the cost of the model from the compressed output (see 
Section \ref{sec:notation}), authors typically resort to zero-th order 
\textit{adaptive} compressors which do not store the symbols'
frequencies, since they are computed {\em incrementally} during the compression \cite{HowardV92}.
A similar approach can be used in this case to achieve the same time 
and space bounds of Lemma \ref{lem:H0}. Here, we require that 
$\len(w_i)=|\C_0 ^a(T[l,r_i])|= |T[l,r_i]|\hka{0}(T[l,r_i])$. 
Recall that with these type of compressors the model must not be stored. 
We use the same tools above but we change the values stored in variables $E_i$ 
and the way in which they
are updated after a $\rem$ or an $\app$.

Observe that in this case we have that 
\[|\C_0 ^a(T[l,r_i])| = |T[l, r_i]| \hka{0}(T[l, r_i]) = \log ((r_i-l+1)!) - \sum _{c \in \Sigma} \log (n_c!)\] 
where $n_c$ is the number of occurrences of symbol $c$ in  $T[l, r_i]$.
Therefore, if the variable $E_i$ stores the value
$\sum _{c \in \Sigma} \log (A_i[c]!)$, then we have that 
$|T[l,r_i]|\hka{0}(T[l,r_i]) = \log ((r_i-l+1)!) - E_i$.\footnote{Notice 
that the value $\log((r_i-l+1)!)$ can be stored in a variable and updated in 
constant time since the size of the value $r_i-l+1$ changes just by one 
after a $\rem$ or an $\app$.}

After the two operations, we change $E$'s value in the following way:

\begin{enumerate}
 \item \rem$()$: For any window $w_i$ we update $E_i$ by subtracting
$\log (A_i[T[l]])$. We also increase $l$ by one.

 \item \app$(w_i)$: We update $E_i$ by summing $\log A[T[r_i+1]]$ and we increase $r_i$ by one.
\end{enumerate}

By the discussion above and Theorem \ref{th:final} we obtain:

\begin{theorem} \label{th:finalH0}
Given a text $T[1,n]$ drawn from an alphabet of size $\sigma = {\tt poly}(n)$,
 we can find an ${(1+\eps)}$-optimal partition of $T$ with
respect to a $0$-th order (adaptive) compressor in $O(n\log_{1+\eps}
n)$ time and $O(n)$ space, where $\eps$ is any positive constant.
\end{theorem}

We point out that these results can be applied to the compression
booster of \cite{FGMS05} to fast obtain an approximation of the
optimal partition of $\BWT(T)$. This may be better than the
algorithm of \cite{FGMS05} both in time complexity, since that
algorithm took $O(n \sigma)$ time, and in compression ratio by a factor 
up to $\Omega(\sqrt{\log n})$ (see the discussion in
Section \ref{sec:intro}). The case of a large alphabet
(namely, $\sigma = \Omega({\tt polylog}(n))$) is particularly
interesting whenever we consider either a word-based \BWT\
\cite{MI05} or the \XBW-transform over labeled trees \cite{FGMS05}.
We notice that our result is interesting also for the {\em Huffword}
compressor which is the standard choice for the storage of Web pages
\cite{Witten:1999:MGC}; here $\Sigma$ consists of the distinct words
constituting the Web-page collection.

\section{On $k$-th order compressors}
\label{sec:hk}

In this section we make one step further and consider the more
powerful $k$-th order compressors, for which do exist $H_k$ bounds
for estimating the size of their compressed output (see Section
\ref{sec:notation}). Here $\len(w_i)$  must compute $|\C(T[l,r_i])|$
which is estimated by $(r_i-l+1) \hk{k}(T[l,r_i]) + f_k(r_i-l+1, |\Sigma_{T[l,r_i]}|)$,
where $\Sigma_{T[l,r_i]}$ denotes the number of different symbols in $T[l,r_i]$..

Let us denote with $T_q[1,n-q]$ the text whose $i$-th symbol $T[i]$
is equal to the $q$-gram $T[i,i+q-1]$. Actually, we can remap the
symbols of $T_q$ to integers in $[1,n]$ without modifying its
zero-th order entropy. In fact the number of distinct $q$-grams
occurring in $T_q$ is less than $n$, the length of $T$. Thus $T_q$'s
symbols take $O(\log n)$ bits and $T_q$ can be stored in $O(n)$
space. This remapping takes linear time and space, whenever $\sigma$
is polynomial in $n$.

A simple calculation shows that the $k$-th order (adaptive) entropy of a string (see
definition Section \ref{sec:notation}) can be expressed as the
difference between the zero-th order (adaptive) entropy of its $k+1$-grams and
its $k$-grams. This suggests that we can use the solution of the
previous section in order to compute the zero-th order entropy of
the appropriate substrings of $T_{k+1}$ and $T_{k}$. More precisely,
we use two instances of the data structure of Theorem
\ref{th:finalH0} (one for  $T_{k+1}$ and one for $T_{k}$), which are
kept {\em synchronized} in the sense that, when operations are
performed on one data structure, then they are also executed on the
other.

\begin{lemma}\label{lem:Hk}
Let $T[1,n]$ be a text drawn from an alphabet of size $\sigma = {\tt poly}(n)$. 
If we estimate $\len()$ via $k$-th order
entropy (as detailed in Section \ref{sec:notation}), then we can design a dynamic data 
structure that takes $O(n)$ space and supports the operations
$\rem$ in $R(n)=O(\log_{1+\eps} n)$ time, and $\app$ and $\len$ in
$L(n)=O(1)$ time.
\end{lemma}

Essentially the same technique is applicable to the case of $k$-th
order {\em adaptive} compressor $\C$, in this
case we keep up-to-date the $0$-th order {\em adaptive} entropies of
the strings $T_{k+1}$ and $T_k$  (details in \cite{ourArxiv}).

\begin{theorem} \label{th:finalHk}
Given a text $T[1,n]$ drawn from an alphabet of size $\sigma = {\tt poly}(n)$, 
we can find an ${(1+\eps)}$-optimal partition of $T$ with
respect to a $k$-th order (adaptive) compressor in
$O(n\log_{1+\eps} n)$ time and $O(n)$ space,  where $\eps$ is any
positive constant.
\end{theorem}

We point out that this result applies also to the practical case in
which the base compressor $\C$ has a maximum (block) size $B$ of
data it can process at once (this is the typical scenario for \gzip,
\bzip, etc.). In this situation the time performance of our solution
reduces to $O(n \log_{1+\eps} (B \log \sigma))$.

\section{On \BWT-based compressors}\label{sec:bwt}

As we mentioned in Section \ref{sec:notation} we know entropy-bounded
estimates for the output size of \BWT-based compressors. So we could
apply Theorem \ref{th:finalHk} to compute the optimal partitioning
of $T$ for such a type of compressors. Nevertheless, it is also
known \cite{FGM06b} that such compression-estimates are rough in
practice because of the features of the compressors that are applied
to the $\BWT(T)$-string. Typically, $\BWT$ is encoded via a sequence
of simple compressors such as $\MTF$ encoding, $\RLE$ encoding
(which is optional), and finally a $0$-order encoder like Huffman or
Arithmetic \cite{Witten:1999:MGC}. For each of these compression
steps, a $0$-th entropy bound is known \cite{Manz01}, but the
combination of these bounds may result much far from the final
compressed size produced by the overall sequence of compressors in
practice \cite{FGM06b}.

In this section, we propose a solution to the optimal partitioning
problem for \BWT-based compressors that introduces a $\Theta(\sigma
\log n)$ slowdown in the time complexity of Theorem
\ref{th:finalHk}, but with the advantage of computing the
$(1+\eps)$-optimal solution wrt the real compressed size, thus
without any estimation by any entropy-cost functions. Since in
practice it is $\sigma= \texttt{polylog} (n)$, this slowdown should
be negligible. In order to achieve this result, we need to address a
slightly different (but yet interesting in practice) problem which
is defined as follows. The input string $T$ has the form
$S[1]\#_1S[2]\#_2 \ldots S[m]\#_n$ where each $S[i]$ is a text
(called {\em page}) drawn from an alphabet $\Sigma$, and $\#_1,\#_2,
\ldots, \#_n$ are special characters greater than any symbol of
$\Sigma$. A partition of $T$ must be page-aligned, that is it must
form {\em groups of contiguous pages} $S[i]\#_i\ldots S[j]\#_j$,
denoted also $S[i,j]$. Our aim is to find a page-aligned partition
whose cost (as defined in Section \ref{sec:intro}) is at most
$(1+\eps)$ the minimum possible cost, for any fixed $\eps>0$. We
notice that this problem generalizes the table partitioning problem
\cite{Tables}, since we can assume that $S[i]$ is a column of the
table.

To simplify things we will drop the $\RLE$ encoding step of a
\BWT-based algorithm, and defer the complete solution to the full
version of this paper. We start by noticing that a close analog of
Theorem~\ref{th:final} holds for this variant of the optimal
partitioning problem, which implies that a $(1+\eps)$-approximation
of the optimum cost (and the corresponding partition) can be
computed using a data structure supporting operations $\app$,
$\rem$, and $\len$; with the only difference that the windows
$w_1,w_2,\ldots,w_m$ subject to the operations are groups of
contiguous pages of the form $w_i = S[l,r_i]$.

It goes without saying that there exist data structures designed to
dynamically maintain a dynamic text compressed with a $\BWT$-based
compressor under insertions and deletions of symbols (see
\cite{permutermTALG} and references therein). But they do not fit
our context for two reasons: (1) their underlying compressor is
significantly different from the scheme above; (2) in the worst
case, they would spend linear space per window yielding a
super-linear overall space complexity.

Instead of keeping a given window $w$ in compressed form, our
approach will only store the frequency distribution of the integers
in the string $w'=\MTF(\BWT(w))$ since this is enough to compute
the compressed output size produced by the final step of the \BWT-based algorithm,
 which is usually implemented via Huffman or Arithmetic \cite{Witten:1999:MGC}. 
Indeed, since $\MTF$ produces a sequence of
integers from $0$ to $\sigma$, we can store their number of
occurrences for each window $w_i$ into an array $F_{w_i}$ of size
$\sigma$. The update of $F_{w_i}$ due to the insertion or the
removal of a page in $w_i$ incurs two main difficulties: (1) how to
update $w_i'$ as pages are added/removed from the extremes of the
window $w_i$, (2) perform this update implicitly over $F_{w_i}$,
because of the space reasons mentioned above. Our solution relies on
two key facts about $\BWT$ and $\MTF$:

\begin{smallenumerate}
\item Since the pages are separated in $T$ by distinct separators,
inserting or removing one page into a window $w$ does not alter the
relative lexicographic order of the original suffixes of $w$ (see
\cite{permutermTALG}).

\item If a string $s'$ is obtained from string $s$ by inserting
or removing a char $c$ into an arbitrary position, then $\MTF(s')$
differs from $\MTF(s)$ in at most $\sigma$ symbols. More precisely,
if $c'$ is the next occurrence in $s$ of the newly inserted (or
removed) symbol $c$, then the \MTF\ has to be updated only in the
first occurrence of each symbol of $\Sigma$ among $c$ and $c'$.
\end{smallenumerate}

Due to space limitations we defer the solution to the Appendix B,
and state here the result we are able to achieve.

\begin{theorem}\label{th:bwt}
Given a sequence of texts of total length $n$ and alphabet size
$\sigma = {\tt poly}(n)$, we can compute an $(1+\eps)$-approximate 
solution to the optimal partitioning problem
for a $\BWT$-based compressor, in $O(n (\log_{1+\eps} n)\; \sigma\; \log n)$ time and $O(n +
\sigma \log_{1+\eps} n)$ space.
\end{theorem}

\section{Conclusion} \label{sec:conclusion}
In this paper we have investigated the problem of partitioning an
input string $T$ in such a way that compressing individually its
parts via a base-compressor $\C$ gets a compressed output that is
shorter than applying $\C$ over the entire $T$ at once. We provide
the first algorithm which is guaranteed to compute in $O(n
\log_{1+\eps}n)$ time a partition of $T$ whose compressed output is
guaranteed to be no more than $(1+\epsilon)$-worse the optimal one,
where $\epsilon$ may be any positive constant. As future directions
of research we would like either to investigate the design of
$o(n^2)$ algorithms for computing the {\em exact} optimal partition,
and/or experiment and engineer our solution over large datasets.

\bibliographystyle{plain}
\bibliography{opartition}

\newpage

\section*{Appendix A\\ An example for the booster}
In this section we prove that there exists an infinite class of
strings for which the partition selected by booster is far from the
optimal one by a factor $\Omega(\sqrt{\log n})$. Consider an
alphabet $\Sigma = \{c_1, c_2,\ldots, c_\sigma \}$ and assume that
$c_1 < c_2 < \ldots < c_\sigma$. We divide it into $l = \sigma/\alpha$
groups of $\alpha$ consecutive symbols each, where $\alpha >0$ will
be defined later. Let $\Sigma_1, \Sigma_2, \ldots, \Sigma_l$ denote
these sub-alphabets. For each $\Sigma_i$, we build a De Bruijn
sequence $T_i$ in which each pair of symbols of $\Sigma_i$ occurs
exactly once. By construction each sequence $T_i$ has length
$\alpha^2$. Then, we define $T=T_1 T_2 \ldots T_l$, so that $|T|=
\sigma \alpha$ and each symbol of $\Sigma$ occurs exactly $\alpha$
times in $T$. Therefore, the first column of $\BWT$ matrix is equal to
$(c_1)^{\alpha} (c_2)^{\alpha} \ldots (c_\sigma)^{\alpha}$. We
denote with $L_c$ the portion of $\BWT(T)$ that has symbol $c$ as
prefix in the $\BWT$ matrix. By construction, if $c \in \Sigma_i$,
we have that any $L_c$ has either one occurrence of each symbol of
$\Sigma_i$ or one occurrence of these symbols of $\Sigma_i$ minus one
plus one occurrence of some symbol of $\Sigma_{i-1}$ (or $\Sigma_l$
if $i=1$). In both cases, each $L_c$ has $\alpha$ symbols, which are
all distinct. Notice that by construction, the longest common prefix
among any two suffixes of $T$ is at most $1$. Therefore, since the
booster can partition only using prefix-close contexts (see
\cite{FGMS05}), there are just three possible partitions: (1) one
substring containing all symbols of $L$, (2) one substring per
$L_c$, or (3) as many substrings as symbols of $L$. Assuming that
the cost of each model is at least $\log \sigma$ bits\footnote{Here
we assume that it contains at least one symbol. Nevertheless, as we
will see, the compression gap between booster's partition and the
optimal one grows as the cost of the model becomes bigger.}, then
the costs of all possible booster's partitions are:
\begin{enumerate}

\item Compressing the whole $L$ at once has cost at least $\sigma \alpha \log \sigma$ bits. In fact, all the symbols
in $\Sigma$ have the same frequency in $L$.

\item Compressing each string $L_c$ costs at least $\alpha \log \alpha + \log \sigma$ bits,
since each $L_c$ contains $\alpha$ distinct symbols.
Thus, the overall cost for this partition is at least $\sigma \alpha \log \alpha + \sigma \log \sigma$ bits.

\item Compressing each symbol separately has overall cost at least $\sigma \alpha \log \sigma$ bits.
\end{enumerate}

We consider the alternative partition which is not achievable by the
booster that subdivides $L$ into $\sigma/ \alpha^2$ substrings
denoted $S_1, S_2, \ldots, S_{\sigma/ \alpha^2}$ of size $\alpha^3$
symbols each (recall that $|T|=\sigma \alpha$). Notice that each
$S_i$ is drawn from an alphabet of size smaller than $\alpha^3$.

The strings $S_i$ are compressed separately. The cost of compressing
each string $S_i$ is $O(\alpha^3 \log \alpha^3 + \log \sigma) =
O(\alpha^3 \log \alpha + \log \sigma)$. Since there are $\sigma/
\alpha^2$ strings $S_i$s, the cost of this partition is $P =
O(\sigma \alpha \log \alpha + (\sigma/\alpha^2) \log \sigma)$.
Therefore, by setting $\alpha = O(\sqrt{\log \sigma}/\log \log
\sigma)$, we have that $P= O(\sigma \sqrt{\log \sigma})$ bits. As
far as the booster is concerned, the best compression is achieved by
its second partition whose cost is $O(\sigma \log \sigma)$ bits.
Therefore, the latter is $\Omega(\sqrt{\log \sigma})$ times larger
than our proposed partition. Since $\sigma \geq \sqrt{n}$, the ratio
among the two partitions is $\Omega(\sqrt{\log n})$.

\section*{Appendix B\\ Proof of Theorem \ref{th:bwt}}
We describe a data structure supporting operations $\app(w)$ and
$\rem()$ when the base compressor is \BWT-based, and the input text
$T$ is the concatenation of a sequence of pages $S[1], S[2], \ldots,
S[m]$ separated by unique separator symbols $\#_1,\#_2,\ldots,\#_m$,
which are not part of $\Sigma$ and are lexicographically larger than
any symbol in $\Sigma$. We assume that the separator symbols in the
$\BWT(T)$ are ignored by the $\MTF$ step, which means that when the
$\MTF$ encoder finds a separator in $\BWT(T)$, this is replaced with
the corresponding integer without altering the $\MTF$-list. This
variant does not introduce any compression penalty (because every
separator occurs just once) but simplifies the discussion that
follows. We denote with $sa_T[1,n]$ and $isa_T[1,n]$ respectively
the suffix array of $T$ and its inverse. Given a range $I=[a,b]$ of
positions of $T$, an occurrence of a symbol of $\BWT(T)$ is called
\textit{active$_{[a,b]}$} if it corresponds to a symbol in $T[a,b]$.
For any range $[a,b] \subset [n]$ of positions in $T$, we define
$\RBWT(T[a,b])$ as the string obtained by concatenating the
active$_{[a,b]}$ symbols of $\BWT(T)$ by preserving their relative
order. In the following, we will not indicate the interval when it
will be clear from the context. Notice that, due to the presence of
separators, $\RBWT(T[a,b])$ coincides with $\BWT(T[a,b])$ when
$T[a,b]$ spans a group of contiguous pages (see \cite{permutermTALG}
and references therein). Moreover, $\MTF(\RBWT(T[a,b]))$ is the
string obtained by performing the \MTF\ algorithm on
$\RBWT(T[a,b])$. We will call the symbol $\MTF(\RBWT(T[a,b]))[i]$ as
the \MTF-encoding of the symbol $\RBWT(T[a,b])[i]$.

For each window $w$, our solution will not explicitly store neither
$\RBWT(w)$ or $\MTF(\RBWT(T[a,b]))$ since this might require a
superlinear amount of space. Instead, we maintain only an array
$F_w$ of size $\sigma$ whose entry $F_w[e]$ keeps the number of
occurrences of the encoding $e$ in $\MTF(\RBWT(w))$. The array $F_w$
is enough to compute the $0$-order entropy of $\MTF(\RBWT(w))$ in
$\sigma$ time (or eventually the exact cost of compressing it with
Huffman in $\sigma \log \sigma$ time).

We are left with showing how to correctly keep updated $F_w$ after a
$\rem()$ or an $\app(w)$. In the following we will concentrate only
on $\app(w)$ since $\rem()$ is symmetrical. The idea underlying the
implementation of $\app(w)$, where $w=S[l,r]$, is to {\em
conceptually insert} the symbols of the next page $S[r+1]$ into
$\RBWT(w)$ one at time from left to right. Since the relative order
among the symbols of $\RBWT(w)$ is preserved in $\BWT(T)$, it is
more convenient to work with active symbols of $\BWT(T)$ by
resorting to a data structure, whose details are given later, which
is able to efficiently answer the following two queries with
parameters $c$, $I$ and $h$, where $c \in \Sigma$, $I=[a,b]$ is a
range of positions in $T$ and $h$ is a position in $\BWT(T)$:

 \begin{itemize}
  \item $\prev _{c}(I,h)$: locate the last active$_{[a,b]}$ occurrence in $\BWT(T)[0,h-1]$ of symbol $c$;
  \item $\next _{c}(I,h)$: locate the first active$_{[a,b]}$ occurrence in $\BWT(T)[h+1,n]$ of symbol $c$.
 \end{itemize}

This data structure is built over the whole text $T$ and requires
$O(|T|)$ space.

Let $c$ be the symbol of $S[r_i+1]$ we have to conceptually insert
in $\RBWT(T[a,b])$. We can compute the position (say, $h$) of this
symbol in $\BWT(T)$ by resorting to the inverse suffix array of $T$.
Once we know position $h$, we have to determine what changes in
$\MTF(\RBWT(w))$ the insertion of $c$ has produced and update $F_w$
accordingly. It is not hard to convince ourselves that the insertion
of symbol $c$ changes no more than $\sigma$ encodings in
$\MTF(\RBWT(w))$. In fact, only the first active occurrence of each
symbol in $\Sigma$ after position $h$ may change its \MTF\ encoding.
More precisely, let $h_p$ and $h_n$ be respectively the last active
occurrence of $c$ before $h$ and the first active occurrence of $c$
after $h$ in $\BWT(w)$, then the first active occurrence of a symbol
after $h$ changes its \MTF\ encoding if and only if it occurs active
both in $\BWT(w)[h_p,h]$ and in $\BWT(w)[h,h_n]$. Otherwise, the new
occurrence of $c$ has no effect on its \MTF\ encoding. Notice that
$h_p$ and $h_n$ can be computed via proper queries $\prev_c$ and
$\next_c$. In order to correctly update $F_w$, we need to recover
for each of the above symbols their old and new encodings. The first
step consists of finding the last active occurrence before $h$ of
each symbols in $\Sigma$ using $\prev$ queries. Once we have these
positions, we can recover the status of the $\MTF$ list, denoted
$\lambda$, before encoding $c$ at position $h$. This is simply
obtained by sorting the symbols ordered by decreasing position. In
the second step, for each distinct symbol that occurs active in
$\BWT(w)[h_p,h]$, we find its first active occurrence in
$\BWT(w)[h,h_n]$. Knowing $\lambda$ and these occurrences sorted by
increasing position, we can simulate the $\MTF$ algorithm to find
the old and new encodings of each of those symbols.

This provides an algorithm to perform $\app(w)$ by making
$O(\sigma)$ queries of types $\prev$ and $\next$ for each symbol of
the page to append in $w$. To complete the proof of the time bounds
in Theorem \ref{th:bwt} we have to show how to support queries of
type $\prev$ and $\next$ in $O(\log n)$ time and $O(n)$ space. This
is achieved by a straightforward reduction to a classic geometric
range-searching problem. Given a set of points $P =
\{(x_1,y_1),(x_2,y_2),\ldots,(x_p,y_p)\}$ from the set
$[n]\times[n]$ (notice that $n$ can be larger than $p$), such that
no pair of points shares the same $x$- or $y$-coordinate, there
exists a data structure \cite{MNlatin06} requiring $O(p)$ space and
supporting the following two queries in $O(\log p)$ time:

\begin{itemize}
 \item {\tt rangemax$([l,r],h)$}: return among the points of $P$ contained in $[l,r] \times [-\infty,h]$ 
the one with maximum $y$-value
 \item {\tt rangemin$([l,r],h)$}: return among the points of $P$ contained in $[l,r] \times [h,+\infty]$ 
the one with minimum $y$-value
\end{itemize}

Initially we compute $isa_T$ and $sa_T$ in $O(n\log \sigma)$ time then, for each symbol $c \in \Sigma$,
we define $P_c$ as the set of points $\{(i,isa_T[i+1])|$ $T[i]=c\}$ and build the above geometric range-searching
structure on $P_c$. It is easy to see that  $\prev _c(I,h)$ can be computed in $O(\log n)$ time by calling {\tt rangemax}$(I,isa_T[h+1])$ on the set $P_c$, and the same holds for $\next _c$ by using {\tt rangemin} instead of
{\tt rangemax}, this completes the reduction and the proof of the theorem.

\end{document}